\documentclass[pra,twocolumn,reprint,amsmath,amssymb,amsfonts,showpacs]{revtex4-1}

\usepackage{graphicx}
\usepackage{amsthm}

\newcommand{\mb}[1]{\ensuremath{\mathbf{#1}}}
\newcommand{\mc}[1]{\ensuremath{\mathcal{#1}}}
\newcommand{\mr}[1]{\ensuremath{\mathrm{#1}}}
\newcommand{\mbb}[1]{\ensuremath{\mathbb{#1}}}

\newcommand{\tr}{\ensuremath{\mathrm{Tr}}}

\newcommand{\la}{\ensuremath{\langle}}
\newcommand{\ra}{\ensuremath{\rangle}}

\newcommand{\bra}[1]{\ensuremath{\langle #1 |}}
\newcommand{\ket}[1]{\ensuremath{| #1 \rangle}}
\newcommand{\prj}[1]{\ensuremath{| #1 \rangle \langle #1 |}}

\newcommand{\matel}[3]{\ensuremath{\langle #1 | #2 | #3 \rangle}}

\newcommand{\ghz}{\ensuremath{\mathrm{GHZ}}}

\newcommand{\g}{\ensuremath{\mathrm{G}}}
\newcommand{\wb}{\ensuremath{\overline{W}}}

\newcommand{\ea}{\textit{et al.}}
\theoremstyle{plain}
\newtheorem{thm}{Theorem}
\newtheorem{lem}{Lemma}

\begin{document}

\title{Minimax optimization of entanglement witness operator for the quantification of three-qubit mixed-state entanglement}
\author{Sungguen Ryu}
\author{Seung-Sup B. Lee}\email{ssblee@kaist.ac.kr}
\author{H.-S. Sim}
\affiliation{Department of Physics, Korea Advanced Institute of Science and Technology, Daejeon 305-701, Korea}

\date{\today}

\begin{abstract}
We develop a numerical approach for quantifying entanglement in mixed quantum states by convex-roof entanglement measures, based on the optimal entanglement witness operator and the minimax optimization method.
Our new approach is applicable to general entanglement measures and states, and it is an efficient alternative to the conventional approach based upon the optimal pure-state decomposition.
Compared with the conventional one, it has two important merits (i) that the global optimality of the solution is quantitatively verifiable, and (ii) that the optimization is considerably simplified by exploiting the common symmetry of the target state and measure.
To demonstrate the merits, we quantify Greenberger-Horne-Zeilinger (GHZ) entanglement in a class of three-qubit full-rank mixed states composed of the GHZ state, the W state, and the white noise, the simplest mixtures of states with different genuine multipartite entanglement, which have not been quantified before this work.
We discuss some general properties of the form of the optimal witness operator, and of the convex structure of mixed states, which are related to the symmetry and the rank of states.
\end{abstract}


\maketitle

\section{Introduction}

Quantum entanglement, a quantum correlation among subsystems, is at the core of strange quantum phenomena, and it is the key ingredient of quantum information processing \cite{Horodecki_RMP}.
Its detection and quantification has attracted much attention of varieties of scientific communities for both fundamental and practical purposes.

Entanglement measure $\mc{E}(\rho)$ quantifies the entanglement contained in a quantum state $\rho$ \cite{EMreviews}.
The definition of $\mc{E}(\rho)$ is not unique and depends on the type of entanglement to be quantified, as there exist many different types of entanglement.
A measure is constructed under the restrictions by the physical meaning of entanglement:
(i) Any measure vanishes for separable states $\rho = \sum_i p_i \, \rho_{\mr{A},i} \otimes \rho_{\mr{B},i} \otimes \cdots$, 
(ii) is invariant under local unitary operations, 
and (iii) does not increase under local quantum operations and classical communications (LOCC).
The last property, so-called monotonicity under LOCC, {is usually treated by adopting} the convexity of a measure, $\mc{E}(p\rho_1 + (1-p)\rho_2) \leq p\,\mc{E}(\rho_1) + (1-p)\mc{E}(\rho_2)$ \cite{Horodecki05}.
A convex-roof measure, the most popular type of measure, is defined first for pure states $\ket{\psi}$, and then extended, based on the convexity, to mixed states $\rho$ via the convex roof construction
\begin{equation}
  \label{convroof}
  \mathcal{E}(\rho) = \inf_{\{p_i , \psi_i \}} \sum_i p_i \; \mathcal{E}(|\psi_i \rangle),
\end{equation}
where the optimization (``inf'') runs over every possible pure state decomposition of $\rho = \sum_i p_i \pi_{\psi_i}$ and $\pi_\psi \equiv \prj{\psi}$ denotes the pure-state projector.

In most cases, except for two-qubit states \cite{Wootters,Schmid,Park},
the quantification of mixed-state entanglement requires heavy computational cost for the optimization of a large number of parameters in Eq. \eqref{convroof} \cite{Roth,Streltsov}.
The pure-state decomposition of a rank-$n$ mixed state into $m$ pure states utilizes a $m \times n$ left-unitary matrix which has $2mn - n^2 -1$ parameters.
As $m$ {ranges} from $n$ to $n^2$ according to Carath\'{e}odory's theorem \cite{Barvinok_ch}, the maximal number of parameters for the optimization becomes $\sim 2n^3$; it already reaches a thousand for three-qubit full-rank states.
The optimization of such a large number of parameters is prone to the convergence to a local optimum, resulting in the overestimation of $\mathcal{E}(\rho)$.
In usual optimization problems, it cannot be ensured, without the overall investigation of the parameter space, whether the result of the optimization is globally optimal, i.e.\ provides the right answer.
These facts make the quantification of multipartite and/or high-dimensional mixed-state entanglement practically unaccessible theoretically and experimentally.

There exists another definition of $\mathcal{E}(\rho)$, which is equivalent to Eq. \eqref{convroof}.
It is the dual problem of Eq. \eqref{convroof} in the sense that it is based on the optimization of quantum operators rather than the optimization of states; see Eq.~\eqref{OWdef} below. The optimized operator, whose expectation value provides  $\mathcal{E}(\rho)$, is called the optimal witness operator \cite{Terhal,Brandao,Eisert}; here the optimal witness operator is used for entanglement quantification rather than detection \cite{Horodecki96,Lewenstein}. 
This concept has been proved to be useful for estimating a lower bound of $\mathcal{E}(\rho)$ with partial information of $\rho$ accessed in experiments \cite{Guhne_PRL}. 

Very recently, it was proposed \cite{Lee12} by two of the authors of the present work that an efficient way of constructing the optimal witness operator can relieve the difficulty of optimizing $\mathcal{E}(\rho)$. The benefits of this approach as an optimization problem are:
(i) The size of the parameter space for this optimization is much smaller than that for Eq.~\eqref{convroof}, and (ii) one can verify whether the result of this approach is globally optimal, i.e.\ provides the true answer, based on the duality between this approach and Eq.~\eqref{convroof}. This approach is generally applicable to any $\mathcal{E}$ and $\rho$.
This approach successfully quantifies three-qubit Greenberger-Horne-Zeilinger (GHZ) entanglement and four-qubit bound entanglement in some full-rank mixed states $\rho$ that have been usually prepared in laboratories, even providing the analytic expression of $\mathcal{E}(\rho)$, demonstrating its usefulness~\cite{Lee12}; 
this was, to our knowledge, the first quantification of entanglement in multipartite full-rank mixed states.

In this paper, we formulate how to numerically obtain the optimal witness operator, utilizing the minimax optimization method \cite{Demyanov}; the utilization is shortly mentioned in Ref.~\cite{Lee12}, but its detail was not investigated. To apply the minimax method, we derive the canonical form of the witness operator that is generally applicable to any measures and states; see Theorem 1 below. We devise the techniques for handling the non-analyticity in the minimax method and of reinforcing the robustness against numerical errors, and also exploit the existing theories of the minimax method \cite{Demyanov,Rustem08} to enhance computation efficiency.
We emphasize that in our approach, the global optimality of a witness operator is quantitatively verified by the minimal distance between the target state and the set of the states exactly quantified by the witness operator. We also notice that the common symmetries of the target state and measure simplify the optimization. 

Our method is applied to scrutinize three-qubit GHZ entanglement in exemplar states. In particular, we quantify GHZ entanglement in a family of full-rank mixed states composed of the GHZ state, the W state, and the white noise, which is the simplest full-rank mixture of pure states with different genuine multipartite entanglement; this type of mixed states has not been quantified before, to our knowledge. This reveals some general properties of the form of the optimal witness operator, and of the convex structure of mixed states, which are related to the symmetry and the rank of the states.

This paper is organized as follows.
In Sec.~\ref{BG}, we introduce the optimal witness operator.
In Sec.~\ref{theory}, we develop the minimax optimization, how to verify the global optimality of its solution, and the simplification utilizing the symmetries.
In Sec.~\ref{result}, we apply our approach to examples, and discuss some general properties of the form of the optimal witness operator.
We conclude in Sec. \ref{conclusion}.

\section{Optimal witness operator}
\label{BG}

We sketch how to quantify mixed-state entanglement by the optimal witness operator, which was developed in Ref. \cite{Lee12}.
The {primal} problem in Eq. \eqref{convroof} is equivalently rephrased as the {dual} problem \cite{Brandao,Lee12} that $\mathcal{E}(\rho)$ is obtained by the optimal operator $X_\rho$,
\begin{equation}
  \begin{aligned}
    \mathcal{E}(\rho) &= \tr(X_\rho \rho) =  \sup_{X \in \mathbb{M}_{\mathcal{H}}} \mathrm{Tr} (X \rho), \\
    \mathbb{M}_{\mathcal{H}} &\equiv \{ X \, | \, \mathrm{Tr} ( X \pi_\psi) \leq
    \mathcal{E}(|\psi\rangle) \;\; \forall|\psi \rangle \in \mathcal{H} \}.
  \end{aligned}
  \label{OWdef}
\end{equation}
Here $\mc{H}$ is a Hilbert space that includes the range of $\rho$, and
any Hermitian operator $X$ satisfies the non-overestimating condition of $\tr(X\rho_1) \leq \mc{E}(\rho_1)$  $\forall \rho_1 \in \mc{H}$ so that $\tr(X \rho)$ is a strict lower bound of $\mc{E}(\rho)$.
$X$ behaves as a {\it generalized} entanglement witness operator:
In the convex structure of quantum states, the hyperplane of $\rho_0$ satisfying $\tr(X\rho_0) = \mc{E}_0$ isolates the convex set of $\{ \rho_1 \, | \, \mc{E}(\rho_1) \leq \mc{E}_0 \}$, namely, $\tr(X \rho_1) \leq \mathcal{E}_0$ for any $\rho_1$ with $\mc{E}(\rho_1) \leq \mc{E}_0$.
Hence $X$ ``witnesses'' entanglement of finite amount $\mathcal{E}_0$, as $\tr(X \rho) > \mc{E}_0$ implies that $\mc{E}(\rho) > \mc{E}_0$.
Note that the conventional entanglement witness \cite{Horodecki96} detects entanglement as it isolates the set of separable states with $\mc{E}_0 = 0$; for mathematical simplicity, in Eq.~\eqref{OWdef}, we use the different convention of a witness operator by the sign factor of $-1$ from the definition in literature.

The optimal entanglement witness operator\cite{Terhal,Brandao,Eisert,Lee12} $X_\rho$ satisfies $\tr(X_\rho \rho) = \mc{E}(\rho)$, quantifying entanglement. It depends on the target state $\rho$ and the choice of $\mathcal{E}$.
From the viewpoint of optimization theory, $X = X_\rho$ is the globally optimal point in $\mbb{M}_\mc{H}$ that maximizes $\tr(X \rho)$.
The global optimum is not necessarily unique for a given $\rho$, and any single global optimum suffices to obtain $\mc{E}(\rho)$.
In some cases, $X_\rho$ does not exist in $\mbb{M}_\mc{H}$, and the supremum in Eq.~\eqref{OWdef} is only asymptotically accessed by an asymptotic form of a witness operator \cite{Lee12}.
The existence of $X_\rho$ depends on the choice of $\mc{H}$.

We mention one important finding of Ref.~\cite{Lee12}.
The duality between Eqs. \eqref{convroof} and \eqref{OWdef} affords the criterion that is used to verify the global optimality of a solution of Eq.~\eqref{OWdef}.
From the fact that $\sum_{i} p_i \matel{\psi_i}{X}{\psi_i}\leq \mc{E}(\rho) \leq \sum_i p_i \mc{E}(\ket{\psi_i})$,
where $\{p_i,\ket{\psi_i} \}$ is a pure state decomposition of $\rho = \sum_i p_i \pi_{\psi_i}$, we have the following equivalence: If $X$ is the optimal witness operator of $\rho$, $\rho$ belongs to the convex hull of $\mathbb{P}_X$, $\mathrm{conv} \, \mathbb{P}_X$.
The converse is also true,
\begin{equation}
X=X_{\rho} \Leftrightarrow \rho \in \mathrm{conv} \, \mathbb{P}_X
\label{GOequiv}
\end{equation}
where $\mbb{P}_{X}$ is the pure-state set of $\{\ket{\psi}\, | \, \tr ( X \pi_\psi ) =\mc{E}(\ket{\psi})  \}$. This implies that we can verify whether a witness operator $X$ is the optimal one for $\rho$, by checking whether $\rho$ is decomposed into a convex combination of the elements of $\mathbb{P}_X$.
Hence, it is possible to check whether a numerical output of the optimization of Eq.~\eqref{OWdef} is at a local optimum (giving only a lower bound of $\mathcal{E}$) or at the global optimum (with the exact value of $\mathcal{E}$). We provide a quantitative way of checking this in Sec. \ref{GO}.

\section{Minimax optimization of witness}
\label{theory}

This section contains the followings. In Sec.~\ref{formulation}, we formulate the minimax optimization of $X_\rho$. In Sec.~\ref{minF}, we devise new techniques to overcome the difficulties that arise in the minimax optimization.
In Sec.~\ref{GO}, based upon the correspondence of Eq. \eqref{GOequiv}, we give how to quantitatively check whether the optimization result is globally optimal.
In Sec.~\ref{symm}, we show that our optimization is simplified when $\mc{E}$ and $\rho$ have common symmetries.
We state the algorithm for the optimization in Sec.~\ref{algorithm}.

\subsection{Formulation}
\label{formulation}

To obtain $X_\rho$, one optimizes the parameters of $X$ under the condition of $\tr(X \pi_\psi) \leq \mc{E}(\ket{\psi})$ in Eq.~\eqref{OWdef}.
To handle this condition efficiently, we translate  it into another optimization problem.
This idea of ``optimization of optimization'' leads to the essential part of our theory:
\begin{thm}
  The optimization of the optimal witness operator $X_\rho$ in Eq. \eqref{OWdef} is formulated as
\begin{equation}
 \begin{aligned}
  \mc{E}(\rho) & = \tr(X_\rho \rho) = \sup_{\Pi \geq 0} \inf_{\ket{\psi} \in \mc{H}} F(\Pi,\ket{\psi}) = \sup_{\Pi \geq 0} G(\Pi),  \\
  & F(\Pi,\ket{\psi}) \equiv \tr ( (\Pi - \mu I) \rho ),  \\
  & \mu(\Pi,\ket{\psi}) \equiv \tr (\Pi \pi_\psi) - \mc{E}(\ket{\psi}),  \\
  & G(\Pi) \equiv \tr((\Pi - \mu_\Pi I) \rho), \,\,\,\,\, \mu_\Pi \equiv \sup_{\ket{\psi} \in \mc{H}} \mu(\Pi,\ket{\psi}).
 \end{aligned}
 \label{minimax}
\end{equation}
\end{thm}
\begin{proof}
  We decompose each Hermitian operator $X \in \mc{H}$ into $X = \Pi - \mu I$ with $\Pi \geq 0$.
  Since $X$ does not overestimate $\mc{E}$, $\mu$ is lower-bounded by $\mu \geq \mu_\Pi$.
  For a given $\Pi$, the finest witness operator is $X = \Pi - \mu_\Pi I$ since a smaller $\mu$ gives a finer $X$.
  Then, the expectation value $\tr X \rho$ is optimized by varying $\Pi$, leading to Eq. \eqref{minimax}.
\end{proof}
\noindent
The minimization (``inf'') directly reflects the non-overestimating condition in Eq.~\eqref{OWdef}.
We note that $X_\rho$ has an asymptotic form, when the supremum in Eq. \eqref{minimax} is asymptotically accessed; in this access,  $\tr((\Pi - \mu_\Pi I)\rho)$ converges, while $\Pi$ and $\mu_\Pi$ diverge.

The optimization in Eq. \eqref{minimax} is a kind of the linear minimax problem \cite{Demyanov}.
The minimax problem has been widely applied to $n$-person games, finance, economics, and policy optimization.
Since it is an optimization problem (``sup'') that includes a sub-optimization (``inf''), it is non-trivial to solve due to the following difficulty.
In a simple-minded approach, one minimizes $F(\Pi,\ket{\psi})$ over $\ket{\psi}$ at every fixed $\Pi$, and then maximizes the resulting $G(\Pi)$ over $\Pi$. 
To apply usual competitive optimization algorithms to the maximization of $G(\Pi)$, one needs to evaluate the first or the higher order derivatives of $G(\Pi)$.
The evaluation of the derivatives requires auxiliary minimization to attain $G(\Pi + d \Pi)$ for small deviation $d \Pi$ in every direction in the domain of $\Pi$.
Then, the computational cost increases enormously as the dimension of the parameter space grows.

Fortunately, there exists an efficient strategy that reduces the computational cost in evaluating the derivatives of $G(\Pi)$.
We here sketch the strategy following Ref. \cite{Demyanov}.
First we define {the solution set} $\mbb{R}(\Pi)$ 
\begin{equation}
  \mbb{R}(\Pi) \equiv \{ \ket{\psi} \in \mc{H} \,|\, F(\Pi,\ket{\psi}) = G(\Pi) \},
\end{equation}
which becomes equivalent to $\mbb{P}_X$ when $X = \Pi - \mu_\Pi I$.
Namely, the elements in $\mbb{R}(\Pi)$ provide the value of $G$ at $\Pi$.
The set $\mbb{R}$ plays a key role in obtaining the derivatives of $G$;
it is also useful for checking the global optimality of the solution of Eq. \eqref{minimax}
as shown in Sec. \ref{GO}.
The behavior of $G$ near $\Pi$ is governed by $\mbb{R}(\Pi)$; for small deviation $d \Pi$, $G(\Pi + d \Pi) = \inf_{\ket{\psi}\in\mbb{R}(\Pi)} F(\Pi + d \Pi,\ket{\psi})$.

When $\mbb{R}$ has multiple elements, the multivariate function $G$ is not differentiable, i.e.\ the derivative of $G$ cannot be expressed by well-defined Jacobian or gradient.
For this case, we use $D_{\delta\Pi}G(\Pi)$, the directional derivative of $G$ along $\delta \Pi$ at $\Pi$, and $D_{\delta \hat{\Pi}}G(\Pi)$, that along the direction of steepest ascent $\delta\hat{\Pi}$ of $G(\Pi)$,
\begin{equation}
  \begin{aligned}
    D_{\delta \Pi}G(\Pi)
    &= \inf_{\ket{\psi} \in \mbb{R}(\Pi)} \tr \left[ (\rho - \pi_\psi \tr \rho) \delta \Pi \right], \\
    D_{\delta \hat{\Pi}}G(\Pi)    
    &= \sup_{\substack{\delta \Pi \in \Gamma(\Pi) \\ \| \delta \Pi \| = 1}} 
    D_{\delta \Pi}G(\Pi)
  \end{aligned}
  \label{ddsd}
\end{equation}
where $\Gamma (\Pi)$ is the cone of possible directions of $\delta \Pi$ at $\Pi$, and the directions are normalized, $\| \delta \Pi \| = \| \delta \hat{\Pi} \| =1$, by the norm of $\|M\|=\sqrt{\tr M^\dagger M}$.

Utilizing Eq. \eqref{ddsd}, $G(\Pi)$ can be maximized by the steepest descent method.
However, there are still two difficulties (i) that the direct evaluation of the steepest ascent using Eq. \eqref{ddsd} is another minimax problem, which may be laborious, and (ii) that the Hessian (the second order derivative) of $G(\Pi)$ may be unavailable when $G(\Pi)$ is not differentiable (i.e.\ when $\mbb{R}(\Pi)$ has multiple elements).
Note that the Hessian is used in most of efficient optimization algorithms since it informs the local behavior of the function to be optimized, which cannot be obtained from the gradient.
For example, the Hessian determines whether a locally optimal solution is stable or unstable by its eigenvalues.
We below devise new techniques to overcome the above difficulties.

\subsection{Minimization of $F(\Pi,\ket{\psi})$}
\label{minF}

To solve the minimax problem in Eq. \eqref{minimax}, we need to accurately solve its ``inf'' problem, avoiding the convergence to local optima.
Below, we propose the techniques that improve the optimization performance in the {inf} problem.
We emphasize that the techniques are applicable to general minimax problems.

{\it Parallelization.---}
The inf problem can be solved by adapting a parallelization approach with multiple solvers.
A solver means a single optimization procedure that finds an (local or global) optimum, following a certain algorithm and starting from an initial guess (randomized or not) of the solution.
We collect the results of the solvers, and assign the best value among them to $G$.
Then $\mbb{R}(\Pi)$ is comprised of $|\psi\rangle$'s providing the best value or the sufficiently good values slightly deviated from the best one within a given tolerance.
$\mbb{R}(\Pi)$ can be fully identified, in principle, given that the unlimited number of solvers are available.
In practical situations, however, parallelization details such as the number of solvers and initial guesses need to be tuned depending on the problem.

To design efficient parallelization, one may require to investigate entanglement classes with respect to a given measure $\mc{E}$.
The entanglement class is an exclusive subset of quantum states whose elements share the common type of entanglement.
For example, there are two entanglement classes for two qubits: one is entangled and the other is separable.
When confined to pure states, some classes, e.g. two-qubit separable states, have vanishing volume.
It is generally possible that some elements of $\mbb{R}(\Pi)$ belong to different entanglement classes, even to classes with zero volume.
However, if a solver is designed to minimize $F(\Pi,\ket{\psi})$ over the full set of pure states at once, the zero-volume classes can be rarely accessed and the solution tends to appear in classes with finite volume.
This can lead to the insufficient construction of $\mbb{R}(\Pi)$.

This problem is resolved by the divide-and-conquer approach:
(i) divide the full set of pure states into the entanglement classes with respect to the target entanglement $\mc{E}$,
(ii) assign solvers to different classes,
and then (iii) let each individual solver minimize $F(\Pi,\ket{\psi})$ over $\ket{\psi}$ only within the assigned class.
In this approach, all the classes can be thoroughly accessed.
Though the full classification of entanglement for more than three qubits is complicated, even an incomplete classification, such as the dichotomy between finite and vanishing $\mc{E}$, improves the performance of the optimization.

{\it Finding elements of $\mbb{R}(\Pi)$.---}
We need sufficient information of $\mathbb{R}$, to precisely evaluate the directional derivative of $G$ or to verify the global optimality of the optimization result.
$\mathbb{R}$ may be numerically found by considering the local optima of the {inf} problem, as mentioned above.
This procedure has two non-trivial points:
(i) The values of $F$ at the local optima may deviate due to numerical errors, and
(ii) all local optima residing in the high-dimensional domain are hard to be found. 

We handle the point (i) by extending the concept of the (exact) solution set $\mathbb{R}$ to the set $\tilde{\mathbb{R}}$ of the candidates for the solutions.
$\tilde{\mathbb{R}}$ consists of not only the global optima but also the local optima which give the value of $F$ sufficiently close to the globally optimal one.
This extension makes $G$ more robust against errors in numerical computation, avoiding the possibility that the global optima are treated, due to the errors, as local optima and excluded from $\mathbb{R}$; the exclusion will induce the failure in evaluating $D_{\delta \Pi}G$, hence the maximization of $G$.

On the point (ii), the parallelization of the {inf} problem is helpful, as it provides a number of local optima, the elements of $\tilde{\mathbb{R}}$.
Sometimes, more elements of $\tilde{\mathbb{R}}$ are required.
They are provided by the local maximum points $\ket{\psi} \in \mc{H}$ of the function
\begin{equation}
  d(\ket{\psi},\ket{\psi_0}) - k\left( F(\Pi,\ket{\psi}) - F(\Pi,\ket{\psi_0}) \right),
  \label{Rset}
\end{equation}
where $d(\ket{\psi},\ket{\psi_0})$ measures the distance between $\ket{\psi}$ and the best optima $\ket{\psi_0}$ obtained from the parallelization, and $k$ is a positive number fixed arbitrarily large.
The exploration of the local maximum points by Eq. \eqref{Rset} depicts the search of $\ket{\psi}$'s that are located far from the best known optimum $\ket{\psi_0}$ but provide the value of $F$ close to the best one;
the second term in the cost function of Eq. \eqref{Rset} acts as the barrier that rules out the points $\ket{\psi}$ providing not-small-enough values of $F$.
With reasonable tolerance for $F$ values, the above search accesses all available candidates of $\tilde{\mathbb{R}}$.
The search can be generalized by considering also the distance from the other locally optimal candidates in $\tilde{\mathbb{R}}$.
Note that if a better optimum point giving a smaller value of $F$ than $\ket{\psi_0}$ appears in this search, then $\ket{\psi_0}$ should be replaced with this better one.
This search complements the parallelization in finding the global optima of the {inf} problem.

{\it Smoothing of $G(\Pi)$.---}
The function $G$ resulting from the {inf} problem can be non-analytic, as mentioned in Sec. \ref{formulation}. 
To surmount the non-analyticity and to efficiently evaluate the derivatives of $G$, we approximate $G$ to an analytic function $\tilde{G}$ within good accuracy.
We call this approximation the ``smoothing'' of $G$.

In this method, we use the candidate set $\tilde{\mathbb{R}} = \{  \ket{\psi_1}, \cdots, \ket{\psi_N} \}$ instead of $\mathbb{R}$, as the computation of $G$ becomes more robust against numerical errors.
When $\tilde{\mbb{R}}$ has enough elements, $G$ is faithfully approximated by $\tilde{G}$,
\begin{equation}
    \tilde{G}(\Pi) \equiv H(F(\Pi,\ket{\psi_1}),\cdots,F(\Pi,\ket{\psi_N}); b).
  \label{approx}
\end{equation}
$H(\alpha_1,\alpha_2,\cdots,\alpha_N; b)$ is an analytic function imitating {$\min \{ \alpha_n \}$} with a smoothing parameter $b > 0$, satisfying 
\begin{equation}
    H(\alpha_1,\cdots,\alpha_N; b) \leq \lim_{b \rightarrow 0} H(\alpha_1,\cdots,\alpha_N; b) =  {\min_{n=1,\cdots,N}} \alpha_n.
  \nonumber
\end{equation}
Larger $b$ makes $H$ smoother, but more deviated from ${\min}_{n=1,\cdots,N} \{ \alpha_n \}$.
As $\tilde{G}(\Pi)$ is a differentiable function with gradient $\nabla \tilde{G}$, the directional derivative in Eq. \eqref{ddsd} is approximately obtained by a simplified form
\begin{equation}
\begin{aligned}
  D_{\delta \Pi} \tilde{G}(\Pi) &=\tr {\left( \nabla \tilde{G} \, \delta \Pi \right)}, \\
  \nabla \tilde{G} &=  \frac{\partial H(\Pi,\{ \ket{\psi_i} \})}{\partial \Pi}.
  \end{aligned}
  \label{ddsd2}
\end{equation}
Similarly, the direction of the steepest ascent is given by $\nabla \tilde{G} / \| \nabla \tilde{G} \|$.
Accordingly, we employ the gradient and the Hessian to maximize $\tilde{G}$ efficiently and precisely, with avoiding the difficulty in directly evaluating Eq.~\eqref{ddsd}.

We mention our choice of the smoothing function $H$,
\begin{equation}
  \begin{aligned}
    H(\alpha_1,\alpha_2;b)& \equiv {\frac{\alpha_1+\alpha_2}{2} - \sqrt{b^2+\frac{(\alpha_1-\alpha_2)^2}{4}}}
     , \\
    H(\alpha_1,\cdots,\alpha_N;b) &\equiv H( H(\alpha_1,\cdots,\alpha_{N-1};b),\alpha_N;b).
  \end{aligned}
  \nonumber
\end{equation}
For $N \geq 3$, $H$ is not exactly symmetric under the permutation of $\alpha_n$, but, the asymmetry is negligible when $b$ is small enough.
Then Eqs. \eqref{approx} and \eqref{ddsd2} are rewritten as
\begin{equation}
  \begin{aligned}
    \tilde{G}(\Pi)  &= \tr \left(\Pi \rho  \right) - \tilde{\mu}_\Pi, \\
    \tilde{\mu}_\Pi & \equiv - H(- \mu(\Pi,\ket{\psi_1}),\cdots,- \mu(\Pi,\ket{\psi_N}) ; b) \geq \mu_\Pi, \\
    \nabla \tilde{G} &= \rho - \frac{\partial \tilde{\mu}_\Pi}{\partial \Pi}.
  \end{aligned}
  \nonumber
\end{equation}
Here, we have used the property of $H(\alpha_1 + \alpha', \cdots, \alpha_N + \alpha';b) = H(\alpha_1, \cdots, \alpha_N;b) + \alpha'$.

\subsection{Verification of global optimality}
\label{GO}

The maximization of $G$ in Eq. \eqref{minimax} provides a candidate of the optimal witness operator. We will show that in our approach, one can verify whether the candidate is the global optimum of the problem, i.e.\ the correct optimal witness operator. 

We first translate the {duality} correspondence in Eq. \eqref{GOequiv} into the following, to use it in the minimax approach.
\begin{thm}
  A witness operator $X = \Pi - \mu_\Pi I$ is optimal for $\mc{E}(\rho)$ if and only if $d_{\min} (\rho;\mr{conv}\,\mbb{R}(\Pi)) = 0$. Here, $d_{\min} (\rho';\mbb{S}) = \min_{\sigma \in \mbb{S}} \| \rho' - \sigma \|_\mr{HS}$ measures the minimal distance between a state $\rho'$ and a set $\mbb{S}$ of states $\sigma$ by the Hilbert-Schmidt norm $\| \cdots \|_\mr{HS}$.
  \label{GOthm}
\end{thm}
\begin{proof}
  $\mbb{P}_X$ is equivalent to $\mbb{R}(\Pi)$ when $X = \Pi - \mu_\Pi I$.
  And $\rho \in \mr{conv}\,\mbb{P}_X$ if and only if $d_{\min} (\rho; \mr{conv}\,\mbb{P}_X) = 0$.
\end{proof}

Based on Theorem~\ref{GOthm}, we use $d_{\min} (\rho;\mr{conv}\,\mbb{R}(\Pi))$ for quantitatively judging how far a candidate $\Pi$ for the solution of Eq. \eqref{minimax} is deviated from the global optimum.
Concretely, $d_{\min} (\rho;\mr{conv}\,\mbb{R}(\Pi))$ is not a distance between the candidate and the global optimum, but a degree of the deviation from the duality correspondence in Eq. \eqref{GOequiv} which should be satisfied by the global optimum.
Hence any {\it a priori} information about the global optimum is not necessary, and the only requirement is that the elements of $\mbb{R}(\Pi)$ are {sufficiently} found by the parallelization of the {inf} problem or by the exploration using Eq. \eqref{Rset}.

To use Theorem \ref{GOthm} in numerical optimization, the extension of $\mbb{R}$ to $\tilde{\mbb{R}}$ and that of $d_{\min} (\rho;\mr{conv}\,\mbb{R}(\Pi))$ to $d_{\min} (\rho;\mr{conv}\,\tilde{\mbb{R}}(\Pi))$ are crucial. When some elements of $\mbb{R}$ (therefore those of $\mr{conv}\,\mbb{R}(\Pi)$) are accidentally excluded due to numerical error, $d_{\min} (\rho;\mr{conv}\,\mbb{R}(\Pi))$ can be overestimated even if $\Pi$ is the globally optimal one.
On the other hand, given an appropriate value of the tolerance for $F(\Pi,\ket{\psi}) - G(\Pi)$, $\tilde{\mbb{R}}(\Pi)$ is the set of all relevant $\ket{\psi}$'s.
The tolerance judges whether a state $\ket{\psi}$ is relevant to $\tilde{\mbb{R}}(\Pi)$ by considering how accurately $\mc{E}(\ket{\psi})$ is quantified by $X = \Pi - \mu_\Pi I$, since $\mc{E}(\ket{\psi}) - \tr(X\pi_\psi) = F(\Pi,\ket{\psi}) - G(\Pi)$ because of $\tr \pi_\psi = \tr \rho = 1$.

The choice of the tolerance for $F(\Pi,\ket{\psi}) - G(\Pi)$ governs the credibility of $d_{\min} (\rho;\mr{conv}\,\tilde{\mbb{R}}(\Pi))$.
Too small tolerance is not helpful, as it causes the problem similar to the case of $\mbb{R}$.
Too large tolerance causes the inclusion of irrelevant elements $\ket{\psi}$'s in $\mr{conv}\,\tilde{\mbb{R}}(\Pi)$,  for which $\tr(X\pi_\psi)$ is far deviated from $\mc{E}(\ket{\psi})$.
It leads to the underestimation of $d_{\min} (\rho;\mr{conv}\,\tilde{\mbb{R}}(\Pi))$, resulting in misleading verification of the global optimality.
In our examples in Sec. \ref{result}, the values of $F(\Pi,\ket{\psi})$ at $\ket{\psi} \in \tilde{\mbb{R}}(\Pi)$ are well-separated from the values at other local optima, hence the adequate value of the tolerance is straightforwardly chosen.

The result of $d_{\min} (\rho;\mr{conv}\,\mbb{R}(\Pi)) = 0$ does not only confirm the global optimality of witness $X = \Pi - \mu_\Pi I$, but also provide the optimal pure-state decomposition of $\rho$ for $\mc{E}$;
a convex sum $\in \mr{conv}\,\mbb{R}(\Pi)$ is the optimal decomposition of $\rho$, when it is identical to $\rho$ and has the vanishing minimal distance of $d_{\rm min}=0$. It is because any convex sum $\rho' \in \mr{conv}\,\mbb{R}(\Pi) = \mr{conv}\,\mbb{P}_X$ itself is the optimal pure-state decomposition by Theorem 2 in Ref. \cite{Lee12}.

We note the possibility that at some $\Pi$'s, $G(\Pi)$'s are {similar} to each other, but $d_{\min} (\rho;\mr{conv}\,\mbb{R}(\Pi))$'s are largely different; see the example in Fig. \ref{fig2}.
{This indicates that $d_{\min} (\rho;\mr{conv}\,\mbb{R}(\Pi))$ is a more useful tool for checking the global optimality than the mere comparison of $G(\Pi)$'s.}

\subsection{Simplification by symmetry}
\label{symm}

The common symmetries of $\mc{E}$ and $\rho$ simplify the minimax approach.
It is identified by the symmetry group $\mbb{G} \equiv \{ U \, | \, U^\dagger = U^{-1}, \, \rho = U \rho U^\dagger, \, \mc{E}(\rho') = \mc{E}(U \rho' U^\dagger) \;\; \forall \rho' \in \mc{H} \}$.
Usually $\mbb{G}$ has local unitary operations and permutations of subsystems; see Sec. \ref{result}.

We first introduce some notions for later use.
We call $O \in \mc{H}$ a symmetric operator with respect to  $\mbb{G}$ when $O = U O U^\dagger$ holds $\forall U \in \mbb{G}$.
We also call $\mbb{C} =  \{P_i \in \mc{H}\}$ a ``basis'' set of symmetric operators $P_i$ with respect to $\mbb{G}$, when $\mbb{C}$ {consists of} the maximal number of linearly independent $P_i$'s satisfying $\tr(P_i P_j) = \delta_{ij}$. 
Any symmetric operator $O$ is decomposed as a linear combination of $P_i$'s, and the coefficients of the decomposition can serve as the parametrization of $O$.  
And, any operator $O \in \mc{H}$ is symmetrized as $O^\mr{S} = \sum_i \tr(O P_i )\, P_i = \sum_{U \in \mbb{G}} U O U^\dagger$; it is obvious that $O^\mr{S} = O$ if and only if $O$ is symmetric.
Note that though the choice of $\mbb{C}$ may not be unique, any choice provides the same result of the symmetrization.

In addition, we  call a set $\mbb{S}$ a symmetric set with respect to $\mbb{G}$ {when} $\mbb{S} = \{ U O U^\dagger \,|\, \forall O \in \mbb{S}, \, \forall U \in \mbb{G} \}$;
accordingly, a set of symmetric operators is symmetric.
There exists a minimal subset $\mbb{S}^\mr{A}$ of a symmetric set $\mbb{S}$ such that $\mbb{S}$ is generated from $\mbb{S}^\mr{A}$ by applying the elements of $\mbb{G}$ to the elements of $\mbb{S}^\mr{A}$. We call $\mbb{S}^\mr{A}$ the asymmetric unit of $\mbb{S}$. It contains the information of $\mbb{S}$ not related to the symmetries. 
The choice of $\mbb{S}^\mr{A}$ may not be unique, but any choice generates {the same} $\mbb{S}$.

Based upon the notions, we find the following features:

\begin{lem}
  There exist the optimal witness operator $X_\rho$ and the corresponding $\Pi_\rho$ symmetric with respect to $\mbb{G}$.
  \label{Xsym}
\end{lem}
\begin{proof}
  Consider the operator of $X = U^\dagger X_\rho U$ for an optimal witness operator $X_\rho$ and any $U \in \mbb{G}$.
  Then $X$ is a witness operator $X \in \mbb{M}_\mc{H}$ not overestimating $\mc{E}$ {since $\tr( U^\dagger X_\rho U \pi_\psi ) = \tr(X_\rho U \pi_{\psi} U^\dagger) \leq \mc{E}(U\ket{\psi}) = \mc{E}(\ket{\psi})$ holds $\forall U \in \mbb{G}$ and $\forall \ket{\psi} \in \mc{H}$}. 
  The witness $X$ is optimal for $\rho$, since $\mc{E}(\rho) = \tr(X_\rho \rho) = \tr(X \rho)$.
  Hence the symmetrization $X_\rho^\mr{S} = \sum_i \tr(X_\rho P_i )\, P_i$ of a non-symmetric optimal witness operator $X_\rho$ constructs the optimal witness operator $X_\rho^\mr{S}$ symmetric under $\forall U \in \mbb{G}$.
  When $X_\rho$ is symmetric, so is $\Pi_\rho$ because $I$ is symmetric.
\end{proof}

\begin{lem}
  If $X$ and $\Pi$ are symmetric with respect to $\mbb{G}$, $\mbb{P}_X$ and $\mbb{R}(\Pi)$ are also symmetric.
  \label{Psym}
\end{lem}
\begin{proof}
  $\forall \ket{\psi} \in \mbb{P}_X$ and $\forall U \in \mbb{G}$, $\mc{E}(\ket{\psi}) = \tr (X \pi_\psi) = \mc{E}(U \ket{\psi}) = \tr(U^\dagger X U \pi_\psi)$. This shows that $U \ket{\psi} \in \mbb{P}_X$, indicating that $\mbb{P}_X$ is symmetric.
  And $\mbb{R}(\Pi)$ is equivalent to $\mbb{P}_X$ for $X = \Pi - \mu_\Pi I$.
\end{proof}

From these two Lemmas, we derive the simplification of the minimax optimization of the optimal witness operator, which is the main result of this subsection:

\begin{thm}
  Suppose $\mc{E}$ and $\rho$ have the common symmetry given by $\mbb{G}$.
  Then in solving Eq. \eqref{minimax}, the domain of the sup problem $\Pi \geq 0$ and the solution set of the inf problem $\mbb{R}(\Pi)$ can be replaced by $\Pi^\mr{S} \geq 0$ and $[\mbb{R}(\Pi^\mr{S})]^\mr{A}$ (with respect to $\mbb{G}$), respectively.
  \label{symthm}
\end{thm}
\begin{proof}
  We prove this theorem by the following two parts.
  
  \noindent
  (i) $\Pi \rightarrow \Pi^\mr{S}$ :
  Lemma \ref{Xsym} ensures the existence of $\Pi_\rho$ symmetric with respect to $\mbb{G}$. Hence the optimum $\mc{E}(\rho)$ is accessible in the sub-domain of symmetric operators $\Pi^\mr{S} \geq 0$.
  It proves the replacement of $\Pi$ by $\Pi^\mr{S}$.

  \noindent
  (ii) $\mbb{R}(\Pi) \rightarrow [\mbb{R}(\Pi^\mr{S})]^\mr{A}$ :
  We start with $\mbb{R}(\Pi^\mr{S})$ that is symmetric with respect to $\mbb{G}$, replacing $\Pi$ by $\Pi^\mr{S}$ as in the part (i).
 We consider a pure state $\ket{\psi} \in \mbb{R}(\Pi^\mr{S})$.
  With small enough $r > 0$, one defines $B_r(\ket{\psi})$, an open ball of radius $r$ centered at $\ket{\psi}$; here the set of pure states  is treated as a metric space with a Hilbert-Schmidt norm metric.
  Then for any $U \in \mbb{G}$, the set of $\{ U \ket{b} \, | \, \forall \ket{b} \in B_r(\ket{\psi}) \}$ becomes $B_r(U \ket{\psi})$.
  Therefore $\mbb{B} \equiv B_r(\ket{\psi_1}) \cup B_r(\ket{\psi_2}) \cup \cdots$ for $\mbb{R}(\Pi^\mr{S}) = \{ \ket{\psi_1}, \ket{\psi_2}, \cdots \}$ is symmetric with respect to $\mbb{G}$.
  Due to the symmetry, $\{ F(\Pi^\mr{S},\ket{\psi}) \,|\, \ket{\psi} \in \mbb{B} \} = \{ F(\Pi^\mr{S},\ket{\psi}) \,|\, \ket{\psi} \in \mbb{B}^\mr{A} \}$ holds.
  It implies that the behavior of $G$ near $\Pi^\mr{S}$ is fully described only by $[\mbb{R}(\Pi^\mr{S})]^\mr{A}$.
\end{proof}

\noindent The replacement in Theorem~\ref{symthm} reduces the computation cost, since a symmetric operator is parametrized by a smaller number of parameters than a non-symmetric one.
Namely, the dimension of the domains for the optimization is reduced by the symmetry.
This reduction makes the optimization easier to converge to the global optima, since local optima (which are deep enough to be stuck) tend to appear less often in a lower dimensional space.

The minimal distance in Theorem~\ref{GOthm}, the quantitative criterion to verify the global optimality of a witness operator, is also streamlined by exploiting the symmetry, i.e.\ by replacing $d_{\min} (\rho;\mr{conv}\,\mbb{R}(\Pi))$ to $d_{\min} (\rho;\mr{conv}\,\mbb{R}(\Pi^\mr{S}))$.

We below apply the symmetries to numerical procedures.
Using a ``basis'' set $\mbb{C} = \{ P_i \}$, we map operators in $\mc{H}$ to real vectors in Euclidean space.
Substituting $\Pi^\mr{S} = \sum_i {v}_i P_i$ and $\rho = \sum_i r_i P_i$ into Eq. \eqref{minimax}, we get
\begin{equation}
  \mc{E}(\rho) = \sup_{\substack{ {\mb{v}} \\ \Pi \geq 0}} \inf_{\ket{\psi} \in \mc{H}} \; \la {\mb{v}},\mb{r} \ra - \left( \la {\mb{v}},\mb{q} \ra -\mc{E}(\ket{\psi}) \right) \tr \rho,
  \nonumber
\end{equation}
where {$\la \cdot, \cdot \ra$ is an inner product of vectors,} $[{\mb{v}}]_i = {v}_i$, $[\mb{r}]_i = r_i$, and $[\mb{q}]_i = \tr(P_i \pi_\psi)$.
Rewriting $F(\Pi,\ket{\psi}) = F({\mb{v}},\mb{q})$ and $G(\Pi) = G({\mb{v}})$, the directional derivative of $G$  at ${\mb{v}}$ in direction ${\mb{v'}}$ ($\|{\mb{v'}}\|=1$) becomes [see Eq. \eqref{ddsd}]
\begin{equation}
  D_{{\mb{v'}}} G = \inf_{\ket{\psi} \in [\mbb{R}(\Pi^\mr{S})]^\mr{A}} \left\la \mb{r} - \tr \rho \,\mb{q}, {\mb{v'}} \right\ra .
  \nonumber
\end{equation}
To smooth $G$, we find the elements of $[\tilde{\mbb{R}}(\Pi^\mr{S})]^\mr{A}$ as the local maximum points $\ket{\psi}\in \mc{H}$ of the function in Eq. \eqref{Rset},
\begin{equation}
  {\| \mb{q} - \mb{q}_0 \| + k( \mu(\Pi^\mr{S},\ket{\psi}) - \mu(\Pi^\mr{S},\ket{\psi_0}))},
  \nonumber
\end{equation}
where $\mu(\Pi^\mr{S},\ket{\psi_0}) = \mu_{\Pi^\mr{S}}$, $[\mb{q}_0]_i = \tr (P_i \pi_{\psi_0})$, {and $k$ is a large positive constant}.
Then the gradient of the smoothed $\tilde{G}$ [see Eq. \eqref{ddsd2}] is 
\begin{equation}
  \nabla \tilde{G} 
  = \mb{r} + \frac{\partial H( - \mu({\mb{v}}, \mb{q}_1), \cdots, - \mu({\mb{v}}, \mb{q}_N) ; b) }{\partial {\mb{v}}},
  \nonumber
\end{equation}
where $\mu ({\mb{v}},\mb{q}_i) = \mu (\Pi,\ket{\psi_i}) = \la {\mb{v}}, \mb{q}_i \ra - \mc{E}(\ket{\psi_i})$ for $[\mb{q}_i]_j = \tr(P_j \pi_{\psi_i})$ and $[\tilde{\mbb{R}} (\Pi^\mr{S})]^\mr{A} = \{ \ket{\psi_1}, \cdots, \ket{\psi_N} \}$.

The minimal distance is also rewritten from the Hilbert-Schmidt norm to the Euclidean norm of $d_{\min} (\mb{r'};\mbb{S}) = \min_{\mb{s} \in \mbb{S}} \| \mb{r'} - \mb{s} \|$.
Then the sufficient and necessary condition of the global optimality becomes
\begin{equation}
  d_{\min} (\mb{r}; \mr{conv}\, \{ \mb{q}_1, \cdots, \mb{q}_N \}) = 0,
\end{equation}
where $[\mb{r}]_i = \tr(P_i \rho)$, $[\mb{q}_i]_j = \tr(P_j \pi_{\psi_i})$, and $\{ \ket{\psi_1}, \cdots, \ket{\psi_N} \} = [\mbb{R} (\Pi^\mr{S})]^\mr{A}$ for a non-smoothed $G$ and $\cdots = [\tilde{\mbb{R}} (\Pi^\mr{S})]^\mr{A}$ for the smoothed $\tilde{G}$.

By finding $\mb{r}' = \sum_i w_i' \mb{q}_i \in \mr{conv}\, \{ \mb{q}_1, \cdots, \mb{q}_N \}$ that provides $d_{\rm min}=0$, we build the optimal pure-state decomposition $\rho = \sum_i w'_i \pi_{\psi_i}^\mr{S}$, where $\pi_{\psi_i}^\mr{S} = \sum_{U \in \mbb{G}} U \pi_{\psi_i} U^\dagger$.
We note that $[\mbb{R} (\Pi^\mr{S})]^\mr{A}$ or $[\tilde{\mbb{R}} (\Pi^\mr{S})]^\mr{A}$ holds the key information of the convex structure of $\mbb{R} (\Pi^\mr{S})$ with respect to $\mc{E}$, hence, is the minimal set for analysis, since $\mbb{R} (\Pi^\mr{S})$ is constructed from $[\mbb{R} (\Pi^\mr{S})]^\mr{A}$ by the symmetry operations.

\subsection{Algorithm}
\label{algorithm}

We summarize this section outlining the entire algorithm for the minimax optimization of witness operator for a general state $\rho$ and a measure $\mc{E}$.
The steps of the algorithm are as follows.

\begin{enumerate}
  \item Investigate the common symmetry of $\mc{E}$ and $\rho$, and construct the corresponding ``basis'' set $\mbb{C} = \{ P_i \}$.
  This step reduces the number of independent parameters. [Sec. \ref{symm}]
  
  \item Inspect the entanglement classes of pure states regarding $\mc{E}$ and divide the set of pure states into the classes.
  It improves the convergence of the problem. [Sec. \ref{minF}]
  
  \item Set the number of parallel solvers and define the smoothing function. [Sec. \ref{minF}]
  
  \item Solve the minimax optimization; maximize $G(\Pi)$ using a competitive optimization algorithm.
  At each $\Pi$, the value and the derivative of $G(\Pi)$ are obtained via the parallelization and the smoothing. [Sec. \ref{minF}]
  
  \item Verify the global optimality by evaluating $d_{\min} (\rho;\mr{conv}\,\mbb{R}(\Pi))$. [Sec. \ref{GO}]
  When the solution is far from the global optimum, increase the number of parallel solvers and/or subdivide the entanglement classes.
  Then retry the optimization.
\end{enumerate}

\section{Three-qubit GHZ entanglement}
\label{result}

We apply the minimax optimization to the quantification of the three-qubit GHZ entanglement in example states.
The example states are composed of three species: The GHZ state {$\ket{\ghz} = \frac{1}{\sqrt{2}}(\ket{000}+\ket{111})$}, the W state {$\ket{W} = \frac{1}{\sqrt{3}}(\ket{100}+\ket{010}+\ket{001})$}, and the white noise $I/8$.
We reproduce the result of Ref. \cite{Lee12} for the mixtures of two out of the three species, and scrutinize the mixture of the three species. The examples reveal some general properties of the form of the optimal witness operator and of the optimal pure-state decomposition, hence, of the convex structure of general mixed states.

This section consists of the introduction of a measure of three-qubit GHZ mixed-state entanglement in Sec.~\ref{ett}, how to apply our minimax approach to the measure in Sec.~\ref{minimaxThreeQ}, three examples of quantification of GHZ mixed-state entanglement in Secs.~\ref{EX1} --~\ref{gwi}, and the usefulness of $d_{\min}$ in the numerical quantification in Sec.~\ref{minimalD}. It also has Sec.~\ref{discussion} where some general properties of the form of the optimal witness operator are discussed.

\subsection{Extensive three-tangle}
\label{ett}

The entanglement structure of three qubits is more complicated than that of two qubits, as there are six classes in the former while only two classes (entangled or not) in the latter.
Among the six classes, there are two genuine tripartite entanglements, GHZ and W, whose representative states are $\ket{\ghz}$ and $\ket{W}$ respectively \cite{3qb}.
The GHZ$\setminus$W class, the set of states with finite GHZ entanglement, comprises the outermost part of the convex structure of three-qubit mixed states \cite{Acin}.

We choose the extensive three-tangle $\mc{T}_3$ \cite{Lee12} as the measure of three-qubit GHZ entanglement.
It quantifies the three-qubit GHZ entanglement and generalizes the three-tangle $\tau_3$ \cite{Coffman} to be invariant under stochastic local operations and classical communications (SLOCC) for mixed states {\cite{Lee12,Gour}}.
It is defined as $\mc{T}_3(\ket{\psi}) = \sqrt{\tau_3(\ket{\psi})}$ for pure states, and extended for mixed states by the convex roof construction.
The name ``extensive'' is given by the property $\mc{T}_3(r \rho) = r \mc{T}_3(\rho)$ for positive real $r$.
The optimal witness operator for $\tau_3$ can be also constructed by Eq. \eqref{minimax}, but we do not discuss it in the present work.

\subsection{Minimax optimization for $\mc{T}_3$}
\label{minimaxThreeQ}

When one applies the minimax formalism in Eq. \eqref{minimax} to $\mc{T}_3$, the domain of pure states for the {inf} problem needs to be carefully considered due to the W class.
For a mixture $\rho$ of GHZ- and W-class states, $\mbb{P}_{X_\rho}$ or $\mbb{R}$ should include the states in both the classes.
However, the set of W-class states has zero volume with respect to pure states, which means that the W-class states are hard to be accessed when a pure state is parametrized to cover the overall pure states.
Even more, the directional derivative of $\mc{T}_3 (\ket{\psi})$ at the W-class pure states in the direction along the GHZ-class pure states shows divergence, {originated from the definition of $\mc{T}_3$; $D_{\ket{\psi}} \mc{T}_3 = (2 \sqrt{\tau_3})^{-1} D_{\ket{\psi}} \tau_3$ diverges when $\mc{T}_3 = \tau_3 = 0$ and $D_{\ket{\psi}} \tau_3$ is finite.}
This divergence obstructs the {inf} problem to be solved correctly.

To avoid the above problems, we separate the domain of the {inf} problem into the GHZ$\setminus$W and W classes, as
\begin{equation}
  \begin{aligned}
    \mc{E}(\rho) &= \sup_{\Pi \geq 0} \, \inf_{\substack{\ket{\psi_\g} \in \ghz \setminus W \\ \ket{\psi_W} \in W}} \tr ( (\Pi - {\tilde{\mu}} I) \rho ), \\
    {\tilde{\mu}} &= H ( \mu_W, \mu_\g ;b) \gtrsim \mu = \max \{ \mu_G, \mu_W \}, \\
    \mu_\g &= \tr (\Pi \pi_{\psi_\g}) - \mc{E}(\ket{\psi_\g}), \\
    \mu_W &= \tr (\Pi \pi_{\psi_W}).
  \end{aligned}
  \nonumber
\end{equation}
Here we parametrize $\ket{\psi_W}$ and $\ket{\psi_\g}$ by using the generalized Schmidt decomposition \cite{Acin00},
\begin{equation}
  \begin{aligned}
    \ket{\psi} = & \; U_1 \otimes U_2 \otimes U_3 \ket{\Psi}, \\
    \ket{\Psi} = & \; \lambda_0 \ket{000} + \lambda_1 e^{i\phi} \ket{100} + \lambda_2 \ket{101} + \lambda_3 \ket{110} \\ & + \lambda_4 \ket{111},
  \end{aligned}
  \nonumber
\end{equation}
where $U_i$'s are local unitary, $\lambda_i \geq 0$, and $\sum_i \lambda_i^2 = 1$.
$\mc{T}_3(\ket{\psi}) = \mc{T}_3(\ket{\Psi}) = 2 \lambda_0 \lambda_4$ holds for GHZ$\setminus$W-class states and $\lambda_4 = \phi = 0$ for W-class states.

We mention the optimization settings chosen for the examples below.
We parallelize the {inf} problem with 30 solvers.
Every solver in the {sup} and {inf} problems runs based on the sequential quadratic programming \cite{SQP} algorithm, an iterative method for nonlinear optimization.

\subsection{Mixture of $\ket{\mr{GHZ}}$ and $I/8$}
\label{EX1}

First our formal approach is applied to the mixed state of $\ket{\mr{GHZ}}$ and the white noise $I/8$, $\rho_{\mr{GI}}(q) = (1-q) \pi_\ghz + q I/8$ with {$q \in (0,1)$}.
The entanglement quantification for this state was analytically studied in Ref. \cite{Lee12}. We reproduce the result of Ref. \cite{Lee12} numerically.

The symmetry group for $\rho_{\mr{GI}}$ consists of (i) the local phase rotation
\begin{equation}
  R_\mr{L}(\alpha,\beta) \equiv (\pi_0 + e^{i\alpha} \pi_1) \otimes (\pi_0 + e^{i\beta} \pi_1) \otimes (\pi_0 + e^{-i(\alpha+\beta)} \pi_1)
  \nonumber
\end{equation}
for the continuous range $\alpha, \beta \in [0, 2\pi)$, (ii) all permutations among three parties, and (iii) 0-1 flip for three parties $\sigma_x ^{\otimes 3}$.
Here $\pi_0 = \prj{0}$, $\pi_1 = \prj{1}$ and $\sigma_x$ is a Pauli matrix.
We select the bases of symmetric operator as $P_{i=0,1} = \Sigma_i / \sqrt{2}$ and $P_2 = (I - \Sigma_0) / \sqrt{6}$, where {we define $\Sigma_{i=0,1,2,3}$ as
\begin{equation}
  \begin{gathered}
    \Sigma_0 \equiv \pi_{000} + \pi_{111}, \quad
    \Sigma_1 \equiv \ket{000}\bra{111} + \ket{111}\bra{000}, \\
    \Sigma_2 \equiv -i\ket{000}\bra{111} + i\ket{111}\bra{000}, \quad
    \Sigma_3 \equiv \pi_{000} - \pi_{111},
  \end{gathered}
  \nonumber
\end{equation}
with $\pi_{000} = \prj{000}$ and $\pi_{111} = \prj{111}$.}

We obtain $\mathcal{T}_3(\rho_{\mr{GI}}(q {\leq q_0}))= 1- q/q_0$ {and $\mathcal{T}_3(\rho_{\mr{GI}}(q > q_0))= 0$} for $q_0 \simeq 0.304$, where the global optimality {of $X_{\rho_\mr{GI}}$} is ensured by $d_{\min} {(\rho;\mr{conv}\,\mbb{R}(\Pi))} \sim 10^{-7}$.
The optimal witness is also found as
\begin{equation}
  X_{\rho_\mr{GI}} \simeq 4.053 P_0 + 1.604 P_1 - 3.000 I \text{ for } q \leq q_0,
  \label{xgi}
\end{equation}
and $X_{\rho_\mr{GI}} = 0$ for $q > 0$.
The optimal decomposition for $q < q_0$ is $\rho_\mathrm{GI}(q) = \mathcal{T}_3 (q) \pi_\mathrm{GHZ} + (1-\mathcal{T}_3(q)) \rho_\mathrm{Z}$, where $\rho_\mathrm{Z}$ is the W-class mixed state
\begin{equation}
  \begin{aligned}
    \rho_\mathrm{Z} &= \frac{1}{18} \sum_{\substack{s=\pm \\ m,n = 0,1,2}} \pi_{\mathrm{Z}_{smn}}, \\
    \ket{\mr{Z}_{smn}} &\equiv \left[ \mathop{\bigotimes}_{\nu = m,n,-m-n}
    \left(
    \begin{smallmatrix}
      1 & 0 \\
      0 & e^{i\frac{2 \nu \pi}{3}}
    \end{smallmatrix}
    \right) \right]
    |\mr{Z}_{s00} \rangle , \\
    | \mathrm{Z}_{+00} \rangle &\equiv \sum_{i,j,k = 0,1} a_{i+j+k} | ijk \rangle , \\
    | \mathrm{Z}_{-00} \rangle &\equiv \sum_{i,j,k = 0,1} a_{3-(i+j+k)} | ijk \rangle ,
  \end{aligned}
  \nonumber
\end{equation}
and $(a_0, a_1, a_2, a_3) \simeq (0.7436, -0.1750, -0.2133, 0.4677)$.
For the above decomposition, the asymmetric unit of $\mbb{P}_{X_{\rho_\mr{GI}}}$ has two elements. We choose them as $\ket{\ghz}$ and $\ket{\mr{Z}_{+00}}$, as the other $\ket{\mr{Z}_{smn}}$'s are obtained from $\ket{\mr{Z}_{+00}}$ by the elements in the symmetry group.

\subsection{Mixture of $\ket{\mr{GHZ}}$ and $\ket{W}$}
\label{EX2}

Next we study $\rho_{\mr{GW}}(p)=(1-p)\pi_\ghz + p \pi_W$ with $p \in (0,1)$, and numerically confirm the analytic result in Ref. \cite{Lee12}.
The form of the optimal witness operator for $\rho_\mr{GW}$ depends on the Hilbert space $\mc{H}$ under consideration. It has an exact form in the range of $\rho_\mr{GW}$, while it has an asymptotic form in the full Hilbert space of three qubits.
Both the forms give the same result of $\mathcal{T}_3(\rho_{\mr{GW}}(p \leq p_0)) = 1- p/p_0$,  $\mathcal{T}_3(\rho_{\mr{GW}}(p > p_0)) = 0$ for $p_0 \simeq 0.3731$, and the optimal pure-state decomposition
\begin{equation}
  \begin{aligned}
    & \rho_\mathrm{GW} (p) = \\
    &
    \begin{cases}
      \mathcal{T}_3(p) \pi_\mathrm{GHZ} + {\displaystyle \frac{1-\mathcal{T}_3(p)}{3} \sum_{n=0,1,2}} \pi_{\mathrm{Z}'_n(p_0)} , & p \le p_0 , \\
      {\displaystyle \frac{1}{3} \sum_{n=0,1,2}} \pi_{\mathrm{Z}'_n(p)} , & p > p_0 ,
    \end{cases}
  \end{aligned}
  \nonumber
\end{equation}
where $|\mathrm{Z}'_n (p) \rangle \equiv \sqrt{1-p} |\mathrm{GHZ} \rangle - \sqrt{p} \; e^{2n \pi i / 3} | \mathrm{W} \rangle$ are W-class states with $\mathrm{Tr}(X_{\rho_\mathrm{GW}} \pi_{\mathrm{Z}'_n}) = 0$.
Without loss of generality, $[\mbb{P}_{X_{\rho_\mr{GW}}}]^\mr{A}$ consists of $\ket{\ghz}$ and $\ket{\mr{Z}'_0(p_0)}$ in the case of $p \leq p_0$, while a single element $\ket{\mr{Z}'_0 (p)}$ for $p > p_0$; the others $\ket{\mr{Z}'_n(p)}$'s with $n = 1, 2$ are obtained from $\ket{\mr{Z}'_0(p_0)}$ by the symmetry operator of $R_\mr{L} (\frac{2n\pi}{3},\frac{2n\pi}{3})$.

By using this example of $\rho_{\mr{GW}}$,
we show how the asymptotic form of the optimal witness operator is handled in the minimax formulation.
In this case where $\mc{H}$ is the full Hilbert space of three qubits, the elements of $\mbb{G}$ are {(i)} discrete local phase rotations $R_\mr{L} (\frac{2n\pi}{3},\frac{2n\pi}{3})$ {with} $n = 0, 1, 2$, and {(ii)} all possible permutations among the three parties.
Note that the symmetry group for {$\rho_\mr{GW}$} is the subgroup of that for $\rho_\mr{GI}$.
The bases of symmetric operator are
\begin{equation}
  \begin{gathered}
    P_{i=0,1,2,3} = \Sigma_i / \sqrt{2}, \\
    P_4 = \pi_{W_1}, \;\; P_5 = \left( \pi_{W_2}+\pi_{W_3} \right)/2, \\
    P_6 = \pi_{\wb_1}, \;\; P_7 = \left( \pi_{\wb_2}+\pi_{\wb_3} \right)/2,
  \end{gathered}
  \label{symcomp}
\end{equation}
where $\ket{W_n}=R_\mr{L}(0,\frac{2(n-1)\pi}{3}) \ket{W}$ and $\ket{\wb_n}=\sigma_x^{\otimes 3} \ket{W_n}$.
To tame the divergent parameters of the optimal witness operator of an asymptotic form, we set the bound of the parameters of $\Pi = \sum_i v_i P_i$ as $|v_i| \leq k$ where $k$ is an arbitrarily large positive number.
The accuracy of the result increases as $k$ grows;
for $k = 10^2$, both of $d_{\min} (\rho;\mr{conv}\,\mbb{R}(\Pi))$ and $\mc{T}_3(\rho_\mr{GW}) - \tr(X_{\rho_\mr{GW}} \rho_\mr{GW})$ are $\sim 10^{-2}$, and for $k = 10^3$ they both decrease to $\sim 10^{-3}$.
We obtain, for $p \leq p_0$, that $\tr(X_{\rho_\mr{GW}} P_{i=0,1}) \,\to\, 1/\sqrt{2}$, $\tr(X_{\rho_\mr{GW}} P_4) \,\to\, 1 - p_0^{-1}$, and the other $\tr(X_{\rho_\mr{GW}} P_{i \neq 0,1,4})$'s diverge, as $k$ increases.
Regardless of the divergence, $\tr(\rho_\mr{GW} X_{\rho_\mr{GW}})$ successfully converges to the exact value of $\mc{E}(\rho_\mr{GW})$ since  $\tr(\rho_\mr{GW} P_{i\neq 0,1,4}) = 0$.

\begin{figure}[bt]
\includegraphics[width=.47\textwidth]{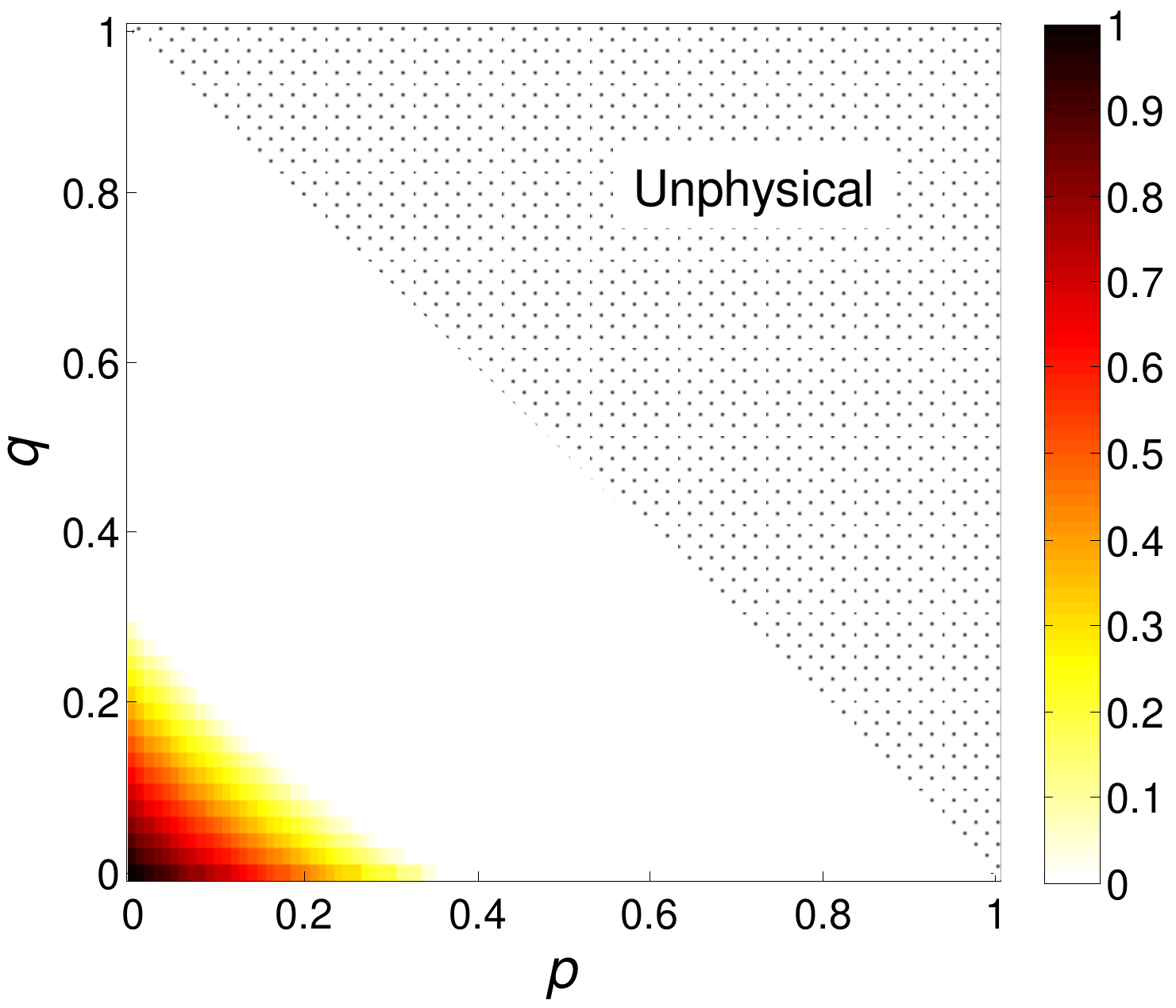}
  \caption{
  (Color online)
  {Extensive three-tangle $\mathcal{T}_3$ of the three-qubit full-rank mixed state of $\rho_{\mr{GWI}}(p,q)=(1-p-q)\pi_\ghz + p \pi_W + q I/8$.}
  The physical region of $(p,q)$ is defined by $p, q, (1-p-q) \in [0,1]$.
  The states on the abscissa and on the ordinate are $\rho_\mr{GW} (p)$ and $\rho_\mr{GI} (q)$ whose values of $\mc{T}_3$ vanish for $p \geq p_0 \simeq 0.373$ and $q \geq q_0 \simeq 0.304$, respectively.
  }
\label{fig1}
\end{figure}

\subsection{Mixture of $\ket{\mr{GHZ}}$, $\ket{W}$, and $I/8$}
\label{gwi}

We consider a more complicated example of the mixture of $\ket{\ghz}$, $\ket{W}$ and $I/8$,
$\rho_{\mr{GWI}}(p,q)=(1-p-q)\pi_\ghz + p \pi_W + q I/8$ with $p,q,p+q \in (0,1)$.
The symmetry of $\rho_{\mr{GWI}}$ is the same as that of $\rho_{\mr{GW}}$; see Eq. \eqref{symcomp}. 
In Fig. \ref{fig1}, we numerically compute $\mc{T}_3 (\rho_\mr{GWI})$ using the minimax optimization.
The global optimality of the optimal witness operator is ensured by $d_{\min} {(\rho;\mr{conv}\,\mbb{R}(\Pi))}  \sim 10^{-4}$ in the computation.

We address the optimal pure-state decomposition of $\rho_{\mr{GWI}}(p,q)$ {for the case that $\rho_{\mr{GWI}}(p,q)$ has finite $\mc{T}_3 (\rho_{\mr{GWI}}(p,q))$}. The decomposition form is numerically obtained, after finding the optimal witness operator.
It consists of one GHZ$\setminus$W-class pure state and multiple W-class pure states as $\rho_\mathrm{GWI} = \alpha \pi_r+(1-\alpha)\rho_Z$, where  $\alpha = \alpha(p,q)$, $\ket{r} = \sqrt{1-r}\ket{000} +\sqrt{r} \ket{111}$ is a GHZ$\setminus$W-class state, $r = r(p,q) \in (0,1)$, and $\rho_Z$ is a mixture of W-class states; this decomposition form indicates that $\mc{T}_{3}$ has the form of $\mc{T}_{3}(\rho_\mathrm{GWI})= \alpha\mc{T}_{3}(\ket{r}) = 2\alpha \sqrt{r(1-r)}$.
$\rho_Z$ is decomposed into pure states in the symmetric set whose asymmetric unit has two elements $\ket{x}$ and $\ket{y_+}$,
  \begin{gather}
    \rho_Z = \beta_1 \rho_{x} + \beta_2 \rho_{y}, \;\; \beta_1 + \beta_2 = 1, \;\; \beta_1, \beta_2 \geq 0, \allowdisplaybreaks[4] \\
    \rho_{x} = \frac{1}{3}\sum_{n=0}^{2} \pi_{x_n}, \;\; \ket{x_n}  =R_\mr{L}({\textstyle \frac{2\pi n}{3},\frac{2\pi n}{3}}) \ket{x}, \nonumber\allowdisplaybreaks[4] \\
    \rho_{y} =\frac{1}{6}\sum_{\substack{m = \pm \\ n=0,1,2}} \pi_{y_{mn}}, \;\; \ket{y_{\pm n}} =R_\mr{L}({\textstyle \frac{2\pi n}{3},\frac{2\pi n}{3}}) \ket{y_{\pm}}, \nonumber \allowdisplaybreaks[4] \\
    \ket{x} =\sum_{i,j,k=0,1} b_{i+j+k}\ket{ijk}, \nonumber \\
    \ket{y_+} = c_0 \ket{000}-c_1 \ket{W_2} -c_2 \ket{\wb_3} + c_3\ket{111}, \nonumber \\
    \ket{y_-} = c_0 \ket{000}-c_1 \ket{W_3} -c_2 \ket{\wb_2} + c_3\ket{111}, \nonumber
  \end{gather}
where $b_i$'s, $c_i$'s and $\beta_i$'s depend on $p$ and $q$.

We describe how the optimal decomposition of $\rho_\mr{GWI}$ is related to those of $\rho_\mr{GI}$ and $\rho_\mr{GW}$.
For $\rho_\mr{GWI}$ with finite $\mc{T}_3$, $[\mbb{P}_{X_\rho}]^\mr{A}$ consists of one GHZ$\setminus$W-class state $\ket{r}$ and two W-class states $\ket{x}$ and $\ket{y_+}$.
On the other hand, $[\mbb{P}_{X_\rho}]^\mr{A}$ contains one GHZ$\setminus$W-class state $\ket{\ghz}$ and one W-class state $\ket{\mr{Z}_{+00}}$ for $\rho_\mr{GI}$ with nonzero $\mc{T}_3$ (i.e.\ $q \leq q_0$), while one GHZ$\setminus$W-class state $\ket{\ghz}$ and another W-class state $\ket{\mr{Z}'_0 (p)}$ for $\rho_\mr{GW}$ with nonzero $\mc{T}_3$ (i.e.\ $p \leq p_0$).
As $\rho_{\mr{GWI}}(p,q)$ continuously approaches to $\rho_\mr{GI}(q)$ by $p \to 0$, we see that $\ket{y_+} \to \ket{\mr{Z}_{+00}}$, $\beta_1 \to 0$, and $\beta_2 \to 1$.
And, $\ket{x} \to \ket{\mr{Z}'_0}$, $\beta_1 \to 1$, and $\beta_2 \to 0$, as $\rho_{\mr{GWI}}(p,q \to 0) \to \rho_\mr{GW}(p)$.

This example of $\rho_{\mr{GWI}}(p,q)$ is the simplest full-rank mixed-state that is composed of pure states with different genuine multipartite entanglement such as GHZ and W. Entanglement in this type of mixed states has not been quantified before, hence, this example demonstrates the usefulness of our numerical approach. 
As discussed in the following subsections, this example demonstrates the usefulness of $d_{\min}$ in Sec.~\ref{minimalD} and reveal some general properties of the form of the optimal witness operator and of the optimal pure-state decomposition, hence, of the convex structure of mixed states in Sec.~\ref{discussion}.

\begin{figure}[t]
  \includegraphics[width=.47\textwidth]{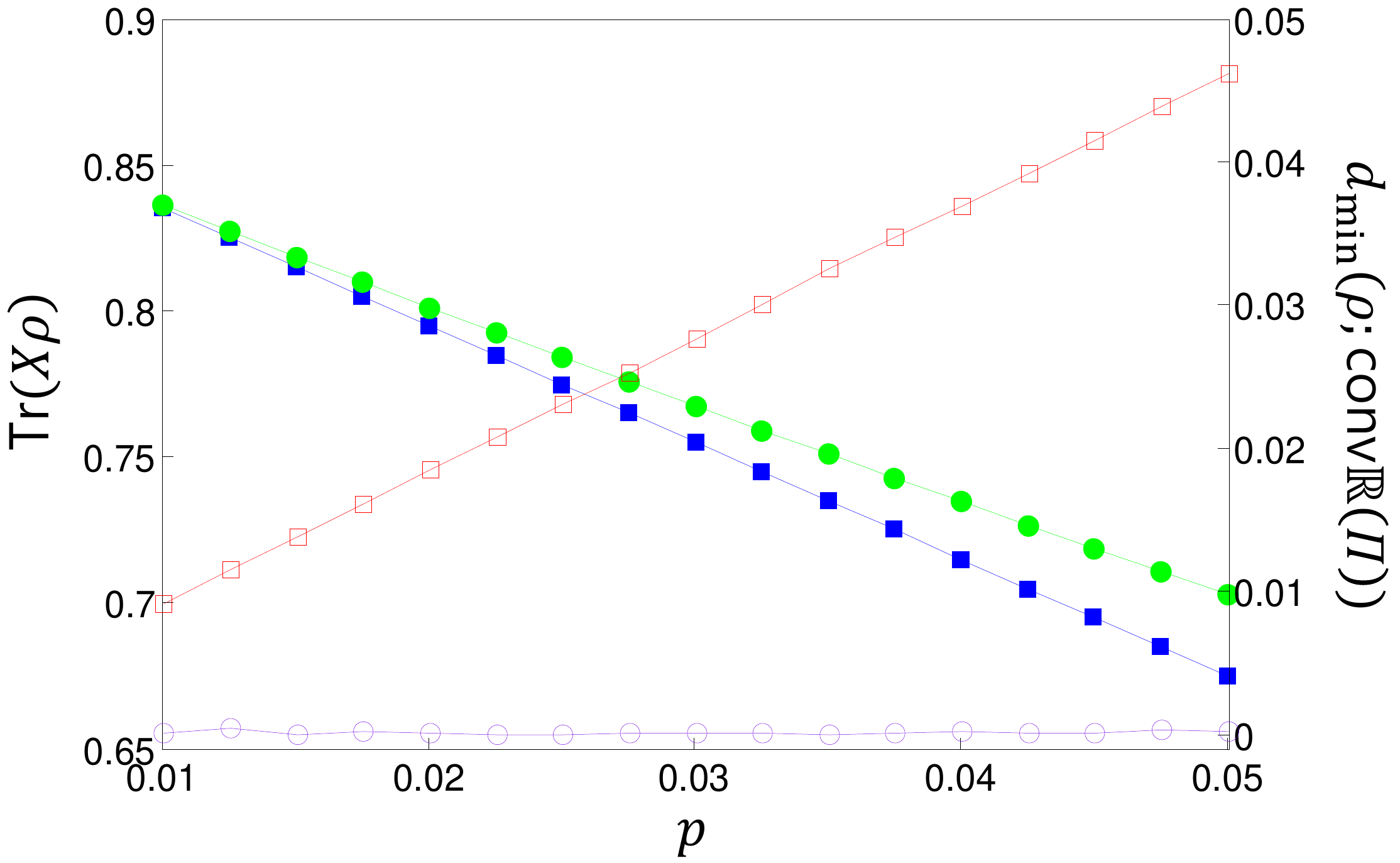}
\caption{(Color online) Quantification of entanglement in $\rho (p) = \rho_{\mr{GWI}}(p,q=0.0380)$ by $\mc{T}_3$. The exact value of $\mc{T}_3 (\rho)$, $\tr(X_\rho \rho)$, is numerically obtained by the optimal witness operator $X_\rho$, and plotted by green filled circles (left ordinate) as a function of $p\in[0.01, 0.05]$. The global optimality of $X_{\rho}$ is verified by the minimal distance $d_{\min} (\rho;\mr{conv}\,\mbb{R}(\Pi))\sim 10^{-4}$, which is also plotted by purple open circles (right ordinate). This exact quantification is compared with the underestimation by a misleading less optimal (locally optimal) witness operator $X_{\rm les}$. The underestimated value of $\mc{T}_3$ by $\tr(X_{\rm les} \rho)$ is plotted by blue filled rectangles. We also plot the minimal distance $d_{\min}$ for this underestimation case by red open rectangles, which clearly shows that $X_{\rm les}$ is not globally optimal. 
}
\label{fig2}
\end{figure}

\begin{figure}
\includegraphics[width=.47\textwidth]{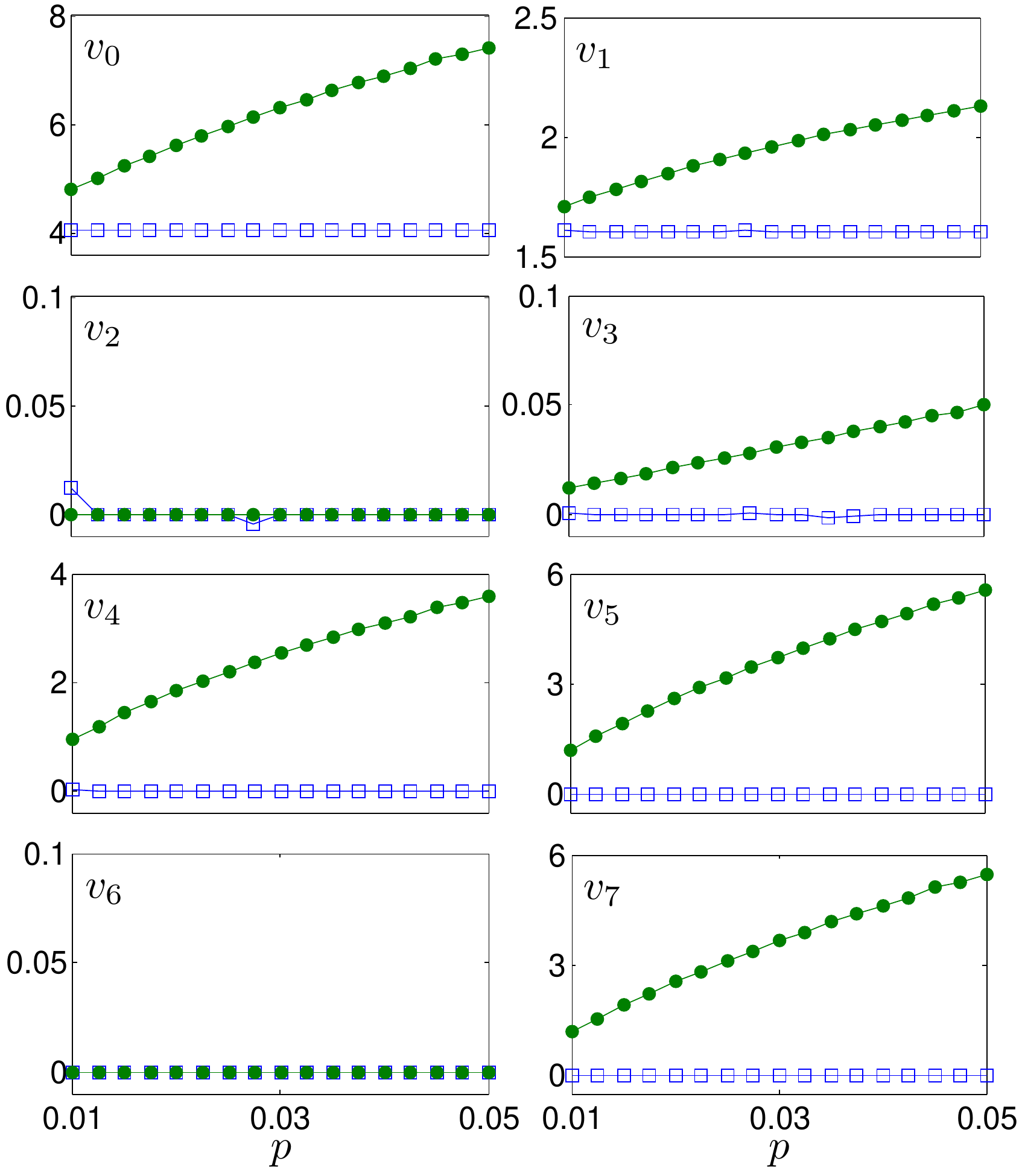}
  \caption{
  (Color online)
Parameters of $\Pi = \sum_i v_i P_i$ of the two witness operators in Fig.~\ref{fig2}, the optimal witness operator $X_\rho$ (green filled circles) and the misleading locally optimal one $X_{\rm les}$ (blue open rectangles), which quantify GHZ entanglement in $\rho_{\mr{GWI}}(p,q=0.0380)$ by $\mc{T}_3$.
We select the bases of symmetric components in Eq. \eqref{symcomp}, since $\rho_\mr{GWI}$ has the same symmetry as $\rho_\mr{GW}$.
  }
\label{fig3}
\end{figure}

\subsection{Usefulness of $d_{\min}$}
\label{minimalD}
 
The example in Sec.~\ref{gwi} shows that the minimal distance $d_{\min} (\rho;\mr{conv}\,\mbb{R}(\Pi))$ is essential for verifying the global optimality.
In Fig. \ref{fig2}, the exact computation of $\mc{T}_3$ by the optimal witness operator $X_\rho (p)$ and a slight underestimation by a {misleading locally} optimal witness operator $X_{\rm les} (p)$ are compared for $\rho_\mr{GWI}(p,q = 0.0380)$; the parameters of the two operators are compared in Fig. \ref{fig3}. The global optimality of the two operators are tested by the minimal distances $d_{\min} (\rho;\mr{conv}\,\mbb{R}(\Pi))$.
The values of $\tr(X\rho)$ by the two witnesses $X_{\rho}$ and $X_{\rm les}$ differ by a small amount ($\sim 10^{-4}$) near $p=0.01$, hence, it is difficult to check the global optimality of a witness operator solely by the expectation vale of $\textrm{Tr} (X \rho)$. In contrast, $d_{\min}$ is clearly distinguishable by $\sim 0.01$ between the two witnesses, even in the case that the values of $\tr(X\rho)$ are very similar. This demonstrates the usefulness of $d_{\min}$ for distinguishing the globally optimal witness operator from other locally optimal (less optimal) ones, namely for testing the global optimality of a witness operator. 

We discuss the appearance of a misleading less optimal witness.
When one searches for $X_\rho$ for general $\rho$ and $\mc{E}$, a misleading less optimal witness $X_\mr{m}$ may frequently appear if the target state $\rho$ neighbors a state $\rho_\mr{h}$ with high symmetry.
In this case, $X_\mr{m}$ shares the high symmetry with $\rho_\mr{h}$.
Since the high symmetry results in many symmetry operators, $\mbb{P}_{X_\mr{m}}$ has lots of elements generated by these operators, and accordingly $\mr{conv}\,\mbb{P}_{X_\mr{m}}$ exhibits large dimension and size.
This indicates that $X_\mr{m}$ is a local optimum point to which the optimization easily converges; any small deviation of a witness operator from $X_\mr{m}$ results in the decrease of $\tr(X \rho)$ since the deviation degrades the symmetry of witness and reduces $\mr{conv}\,\mbb{P}_{X_\mr{m}}$ significantly.
In the example of $\rho_\mr{GWI}$ in Sec. \ref{gwi}, the misleading witness operator $X_\mr{les}$ has the form similar to $X_{\rho_\mr{GI}}$; one can notice this by comparing the parameters of $X_\mr{les}$ shown in Fig. \ref{fig3} with Eq. \eqref{xgi}.
In the parametrization space with respect to Eq. \eqref{symcomp}, the dimension of $\mr{conv}\,\mbb{P}_{X_{\rho_\mr{GI}}}$ is six while $\mr{conv}\,\mbb{P}_{X_{\rho_\mr{GWI}}}$ is only three-dimensional.
Due to this difference, $X_\mr{les}$ may appear frequently though it is mere local optimum.
Such misleading local optimal witnesses can be ruled out by computing $d_{\min}$.

\subsection{Perturbative construction of optimal witness}
\label{discussion}

When one constructs the optimal witness operator for $\rho'=(1-\delta)\rho +\delta \sigma $ ($\delta \ll 1$), one may guess that the optimal witness operator for $\rho$ is useful, and try to construct the optimal witness operator for $\rho'$ from that for $\rho$;
we call this approach the {\it perturbative construction} of the optimal witness operator.
This idea is based upon the continuity of entanglement measure, $| \mc{E}(\rho') - \mc{E}(\rho) | \rightarrow 0$ as $\delta \rightarrow 0$, and upon the assumption that the form of the optimal witness operator or the form of the optimal pure-state decomposition changes continuously and smoothly from $\rho$ to $\rho'$.
This perturbative approach works well for the states inside each family of $\rho_\mr{GI}$, $\rho_\mr{GW}$, and $\rho_\mr{GWI}$.
The success of the perturbative construction saves the number of parallelized solvers for the {inf} problem in Eq. \eqref{minimax}.

However, the perturbative construction is not always successful. For example, one may try to perturbatively construct the optimal witness operator for $\rho_\mr{GWI}$, starting from that for $\rho_\mr{GI}$ or $\rho_\mr{GW}$, i.e., from the simpler case for which the optimal witness operator is known.
It fails, because of the change of the symmetry (in the case from $\rho_\mr{GI}$) or the change of the rank (from $\rho_\mr{GW}$).
In the case from $\rho_\mr{GI}$, the witness for $\rho_\mr{GWI}$ perturbatively constructed from $X_{\rho_\mr{GI}}$ is $X_{\rm les}$. As shown in Figs. \ref{fig2} and \ref{fig3} (see rectangles), $X_{\rm les}(p)$ underestimates $\mc{T}_3 (\rho_\mr{GWI}(p,q = 0.0380))$, is not globally optimal, and has a form different from $X_\rho (p)$, indicating the failure of the perturbative approach.
In the other case from $\rho_\mr{GW}$, the parameters of the full-rank witness operator for $\rho_\mr{GWI}$ diverge at $\rho_\mr{GW}$ (i.e.\ at $q=0$); note that as discussed before, the optimal witness operator for $\rho_\mr{GW}$ has an asymptotic form, when one choose $\mc{H}$ as the full Hilbert space of three qubits.

The failure of the perturbative construction is manifested by the non-smooth evolution of the convex structure (with respect to $\mc{E}$) from $\rho$ to $\rho'$.
The non-smooth evolution means that the optimal witness $X_\rho$ or the convex set $\mr{conv}\,\mbb{P}_{X_\rho}$ evolves non-smoothly between the two cases.
We notice this feature, by observing the form of the optimal pure-state decomposition in Secs.~\ref{EX1} --~\ref{gwi}, that $\mbb{P}_{X_\rho}$ suddenly changes from $\rho_\mr{GWI}$ to $\rho_\mr{GI}$ or $\rho_\mr{GW}$.
In the case of $\rho_\mr{GI}$, $\mbb{P}_{X_\rho}$ consists of infinitely many states:
One of them is $\ket{\ghz}$ and the others are the states generated from $\ket{Z_{+00}}$ by the symmetric operators of $\rho_\mr{GI}$, including the continuous local phase rotation $R_\mr{L} (\alpha,\beta)$ with $\forall \alpha,\beta \in [0, 2\pi)$.
In contrast, there are ten elements $\in \mbb{P}_{X_\rho}$ in the case of $\rho_\mr{GWI}$, and only four for $\rho_\mr{GW}$.
Hence,
the convex structure of $\mr{conv}\,\mbb{P}_{X_\rho}$ is non-smoothly connected among $\rho_\mr{GWI}$, $\rho_\mr{GI}$, and $\rho_\mr{GW}$.

This example shows that there is no single universal form of the optimal witness operator in the three-qubit case, in contrast to the two-qubit case \cite{Park}, and that the convex set ${\mr{conv}\,}\mbb{P}_{X_\rho}$ (or its evolution) needs to be studied to construct the optimal witness operator and to understand multipartite or high-dimensional entanglement.

\section{Conclusion}
\label{conclusion}

We develop a numerical optimization approach for finding the optimal witness operator that quantifies mixed-state entanglement.
It is applicable to general entanglement measures and states. Importantly, the global optimality of the optimization result is verifiable.
Further development of our approach will provide a powerful tool for studying not only multipartite mixed-state entanglement but also other convexity-related physical properties such as nonlocality and quantum channels.

\section*{Acknowledgement}

We thank NRF for support (Grant No. 2009-0084606).

\end{document}